\documentclass[
a4paper,reqno]{amsart}
\usepackage[T1]{fontenc}
\usepackage[english]{babel}
\usepackage{amssymb,bbm,enumerate}
\usepackage{graphicx}
\usepackage{float} 
\restylefloat{table}

\usepackage{bm}
\usepackage{xcolor}

\usepackage[linktocpage=true,colorlinks=true, linkcolor=blue, citecolor=red, urlcolor=green]{hyperref}

\newcommand{\mc}[1]{\mathcal{#1}}
\newcommand{\mf}[1]{\mathfrak{#1}}
\newcommand{\mb}[1]{\mathbb{#1}}

\newcommand{\tint}{{\textstyle\int}}
\newcommand{\git}{\mathbin{/\mkern-6mu/}}

\renewcommand{\ker}{\Ker}

\DeclareMathOperator{\End}{End}

\DeclareMathOperator{\ad}{ad}
\DeclareMathOperator{\im}{Im}

\DeclareMathOperator{\Ker}{Ker}
\DeclareMathOperator{\Span}{Span}

\DeclareMathOperator{\st}{st}

\DeclareMathOperator{\Spec}{Spec}

\theoremstyle{plain}
\newtheorem{theorem}{Theorem}[section]
\newtheorem{lemma}[theorem]{Lemma}
\newtheorem{proposition}[theorem]{Proposition}

\theoremstyle{definition}
\newtheorem{definition}[theorem]{Definition}
\newtheorem{example}[theorem]{Example}

\theoremstyle{remark}
\newtheorem{remark}[theorem]{Remark}

\numberwithin{equation}{section}

\definecolor{light}{gray}{.9}

\usepackage{bigstrut,tikz}

\usetikzlibrary{arrows,arrows.meta,calc,decorations,matrix,overlay-beamer-styles,positioning,through}

\tikzset{%
    add/.style args={#1 and #2}{
        to path={%
 ($(\tikztostart)!-#1!(\tikztotarget)$)--($(\tikztotarget)!-#2!(\tikztostart)$)%
  \tikztonodes},add/.default={.2 and .2}}
}

\pgfdeclaredecoration{arrows}{draw}{
\state{draw}[width=\pgfdecoratedinputsegmentlength]{%
  \path [every arrow subpath/.try] \pgfextra{%
    \pgfpathmoveto{\pgfpointdecoratedinputsegmentfirst}%
    \pgfpathlineto{\pgfpointdecoratedinputsegmentlast}%
   };
}}

\tikzset{every arrow subpath/.style={->, draw, thick}}

\newcommand{\foresix}[8]{
\begin{tikzpicture}[baseline=(a.base),scale=#1,every node/.style={scale=#2},inner sep=0pt,outer sep=0pt]
\node (a) at (1,0) {$#8$};
\node at (2,0) {$#7$};
\node at (3,0) {$#6$};
\node at (3,-1.5) {$#4$};
\node at (4,0) {$#5$};
\node at (5,0) {$#3$};
\end{tikzpicture}
}

\newcommand{\foreseven}[9]{
\begin{tikzpicture}[baseline=(a.base),scale=#1,every node/.style={scale=#2},inner sep=0pt,outer sep=0pt]
\node (a) at (1,0) {$#9$};
\node at (2,0) {$#8$};
\node at (3,0) {$#7$};
\node at (4,0) {$#6$};
\node at (4,-1.5) {$#4$};
\node at (5,0) {$#5$};
\node at (6,0) {$#3$};
\end{tikzpicture}
}

\newcommand{\foreight}[9]{
\begin{tikzpicture}[baseline=(a.base),scale=.1,every node/.style={scale=#1},inner sep=0pt,outer sep=0pt]
\node (a) at (1,0) {$#9$};
\node at (2,0) {$#8$};
\node at (3,0) {$#7$};
\node at (4,0) {$#6$};
\node at (5,0) {$#5$};
\node at (5,-1.5) {$#3$};
\node at (6,0) {$#4$};
\node at (7,0) {$#2$};
\end{tikzpicture}
}

\newlength{\halogap}

\newlength\triplesep
\newlength\triplelinewidth
\tikzset{triple/.style={
line width=\triplelinewidth,
black,
 preaction={
  preaction={
    draw,
    line width=2\triplesep+3\triplelinewidth,
    black
   },
   draw,
   line width=2\triplesep+\triplelinewidth,white
  }
 }
}

\tikzset{DynkinNode/.style={circle,minimum size=.7em,inner sep=0pt,font=\scriptsize}}

\newcommand{\GDynkin}[1]{
\begin{tikzpicture}[scale=.8,anchor=base,baseline]
\foreach\kthweight[count=\k] in {#1}{
\ifnum\k=1\node[DynkinNode] (1) at (-1,0) {$\scriptscriptstyle\kthweight$};\fi
\ifnum\k=2\node[DynkinNode] (2) at (-2,0) {$\scriptscriptstyle\kthweight$};\fi
}
\draw[-Implies,double distance=1pt] (2) -- (1);
\draw (2) -- (1);
\end{tikzpicture}
}

\newcommand{\FDynkin}[1]{
\begin{tikzpicture}[scale=.6,anchor=base,baseline]
\foreach\kthweight[count=\k] in {#1}{
\ifnum\k=1\node[DynkinNode] (1) at (-1,0) {$\scriptscriptstyle\kthweight$};\fi
\ifnum\k=2\node[DynkinNode] (4) at (-4,0) {$\scriptscriptstyle\kthweight$};\fi
\ifnum\k=3\node[DynkinNode] (2) at (-2,0) {$\scriptscriptstyle\kthweight$};\fi
\ifnum\k=4\node[DynkinNode] (3) at (-3,0) {$\scriptscriptstyle\kthweight$};\fi
}
\draw[thick] (1) -- (2);
\draw[-Implies,double,thick] (3) -- (2);
\draw[thick] (3) -- (4);
\end{tikzpicture}
}

\newcommand{\EDynkin}[2]{
\begin{tikzpicture}[scale=.4,anchor=base,baseline]
\foreach\kthweight[count=\k] in {#1}{
\pgfmathtruncatemacro{\prevnode}{\k-1}
\ifnum\k=1\node[DynkinNode] (\k) at (0,-1) {$\scriptscriptstyle\kthweight$};\fi 
\ifnum\k>1\node at (4-\k,0) {$\scriptscriptstyle\kthweight$};
\node[DynkinNode] (\k) at (4-\k,0) {$\scriptscriptstyle\kthweight$};
\ifnum\k>2\draw[thick] (\k) -- (\prevnode);\fi
\ifnum\k=4\draw[thick] (\k) -- (1);\fi
\fi
}
\foreach\Node[count=\i] in {#2}
{
\xdef\imax{\i}
\coordinate (AuxNode-\i) at ($(\Node)$);
}
\ifnum\imax=3%
\draw[red,rounded corners] ($(AuxNode-1)+(-\halogap,\halogap)$) -|  ($(AuxNode-2)+(-\halogap,\halogap)$)
    -- ($(AuxNode-2)+(\halogap,\halogap)$) |-
    ($(AuxNode-3)+(\halogap,\halogap)$)
    --($(AuxNode-3)+(\halogap,-\halogap)$) -- ($(AuxNode-1)+(-\halogap,-\halogap)$) -- cycle;
\else\ifnum\imax=2%
\draw[red,rounded corners] ($(AuxNode-1)+(-\halogap,\halogap)$) --  ($(AuxNode-2)+(\halogap,\halogap)$)
    -- ($(AuxNode-2)+(\halogap,-\halogap)$) -- ($(AuxNode-1)+(-\halogap,-\halogap)$)--
    cycle;
\fi
\fi
\end{tikzpicture}
}

\begin{document}

\title{Integrability of classical affine $W$-algebras}

\author{Alberto De Sole}
\address{Dipartimento di Matematica, Sapienza Universit\`a di Roma,
P.le Aldo Moro 2, 00185 Rome, Italy}
\email{desole@mat.uniroma1.it}
\urladdr{www1.mat.uniroma1.it/\~{}desole}
\author{Mamuka Jibladze}
\address{Razmadze Mathematical Institute,
TSU, Tbilisi 0186, Georgia
}
\email{jib@rmi.ge}
\author{Victor G. Kac}
\address{Dept of Mathematics, MIT,
77 Massachusetts Avenue, Cambridge, MA 02139, USA}
\email{kac@math.mit.edu}
\author{Daniele Valeri}
\address{School of Mathematics and Statistics, University of Glasgow, G12 8QQ Glasgow, UK}
\email{daniele.valeri@glasgow.ac.uk}



\begin{abstract}
We prove that all classical affine $W$-algebras $\mc W(\mf g,f)$, where
$\mf g$ is a simple Lie algebra and $f$ is its non-zero nilpotent element,
admit an integrable hierarchy of bi-Hamiltonian PDEs, except possibly
for one nilpotent conjugacy class in $G_2$, one in $F_4$, and five in $E_8$.
\end{abstract}

%

\maketitle

\begin{center}
\emph{To the memory of Ernest Borisovich Vinberg}
\end{center}

\tableofcontents

\section{Introduction}\label{sec:1}
In order to define a Hamiltonian ODE one needs a Poisson algebra, that is a commutative associative algebra $\mc P$, endowed with
a Poisson bracket $\mc P\otimes\mc P\to\mc P$, $a\otimes b\mapsto\{a,b\}$, and an element $h\in\mc P$, called a
Hamiltonian function. Then the ODE
$$
\frac{du}{dt}=\{h,u\}\,,\quad u=u(t)\in\mc P\,,
$$
is called Hamiltonian.
Recall that, by definition of a Poisson algebra, the bracket $\{\cdot\,,\,\cdot\}$ should satisfy the Lie algebra axioms, and, for every $f\in\mc P$,  the map
$u\mapsto\{f,u\}$ should be a derivation of the associative product on $\mc P$ (Leibniz rule).

In a similar fashion, in order to define a Hamiltonian PDE one needs a \emph{Poisson vertex algebra} (abbreviated PVA),
which is a differential algebra, that is a commutative associative algebra $\mc V$ with a derivation $\partial$, endowed with
a \emph{PVA $\lambda$-bracket} $\mc V\otimes\mc V\to \mc V$, $a\otimes b\mapsto\{a_\lambda b\}$,
and an element $\tint h\in\mc V/\partial\mc V$, called a \emph{local Hamiltonian functional}.
(One denotes by $\int$ the canonical map $\mc V\to\mc V/\partial\mc V$ since it is the universal map
satisfying integration by parts: $\tint (\partial f)g=-\tint f\partial g$.)
A key property of a PVA $\mc V$ is that the formula
\begin{equation}\label{intro:eq2}
\{\tint f,\tint g\}=\tint \{f_\lambda g\}|_{\lambda=0}
\end{equation}
produces a well-defined Lie algebra bracket on $\mc V/\partial\mc V$.
The PDE
\begin{equation}\label{intro:eq1}
\frac{du}{dt} =
\{ \tint h , u\}
:=
\{ h_\lambda u\}
\big|_{\lambda=0}\,,\quad u\in\mc V\,,
\end{equation}
is called Hamiltonian.
In \eqref{intro:eq1} $u=u(x,t)$ should be viewed as a function in the space variable $x$ and the time $t$,
$\partial=\frac{\partial}{\partial x}$ is the partial derivative with respect to $x$,
while $\frac{d}{dt}$ defines the time flow of the system.
Recall that, by definition, the PVA $\lambda$-bracket should satisfy the Lie conformal algebra axioms, see axioms (i)-(iii)
from Section \ref{sec:2.2}, similar to the Lie algebra axioms, and the Leibniz rules (iv) and (iv') from Section \ref{sec:2.2}.
In particular, due to the first sesquilinearity axiom (i), the RHS of equation \eqref{intro:eq1}
is well-defined (i.e. it does not depend on the choice of the representative of the coset $\tint h$).

Recall that if $\mc V$ carries two compatible PVA $\lambda$-brackets
(i.e. any their linear combination is a PVA $\lambda$-bracket), such that
an evolution PDE can be written in the form \eqref{intro:eq1} for both of them, this
evolution equation is called bi-Hamiltonian.

The notion of a PVA appears naturally in the theory of vertex algebras as their quasiclassical limit, cf. \cite{DSK06},
in the same way as a Poisson algebra appears naturally as a quasiclassical limit of a family of associative algebras.

However, the theory of Hamiltonian PDEs was started 15 years before the advent of the vertex algebra theory,
in the work of Faddeev and Zakharov \cite{ZF71}, who attribute the construction to C. S. Gardner, which, in algebraic terms,
is as follows.
Let $\mc V=\mb F[u_i^{(n)}\mid i\in I=\{1,\dots,\ell\}, n\in\mb Z_{\geq 0}]$ be the algebra of differential polynomials in $\ell$ differential
variables, with the derivation $\partial$ defined by $\partial u_i^{(n)}=u_i^{(n+1)}$. Given $\tint h\in\mc V/\partial\mc V$,
one defines a Hamiltonian PDE
\begin{equation}\label{intro:eq3}
\frac{du}{dt}=H(\partial)\frac{\delta }{\delta u}\tint h\,,
\end{equation}
where $u=(u_i)_{i\in I}$ is an $\ell$-column vector of dependent variables,
$H(\partial)$ is an $\ell\times\ell$-matrix differential operator with coefficients
in $\mc V$, called the \emph{Poisson structure} of $\mc V$,
and $\frac{\delta}{\delta u}\tint h=\big(\frac{\delta}{\delta u_i}\tint h\big)_{i\in I}$ is the $\ell$-column vector
of variational derivatives
$$
\frac{\delta}{\delta u_i}\tint h=\sum_{n\in\mb Z_{\geq 0}}(-\partial)^n\frac{\partial h}{\partial u_i^{(n)}}
\,\,,\,\,\,\,
i\in I
\,.
$$
(A Hamiltonian ODE can be written in local coordinates in a similar form $\frac{du}{dt}=H\frac{\partial }{\partial u}h$,
where $H$ is the transpose of the matrix of Poisson brackets of coordinate functions $u_i$, and $\frac{\partial}{\partial u}$ is the gradient.)

The basic assumption on the Poisson structure $H(\partial)$ is that formula
\begin{equation}\label{intro:eq4}
\{\tint f,\tint g\}=\int \frac{\delta \int g}{\delta u}H(\partial)\frac{\delta \int f}{\delta u}
\end{equation}
defines a Lie algebra structure on $\mc V/\partial\mc V$.
A simple observation is that equation \eqref{intro:eq3} coincides with \eqref{intro:eq1},
and \eqref{intro:eq4} coincides with \eqref{intro:eq2},
if we let $H(\partial)=\left(\{{u_j}_\partial u_i\}_\rightarrow\right)_{i,j\in I}$,
where the arrow means that $\partial$ should be moved to the right, see \cite{BDSK09}.
Thus PVAs provide a coordinate free approach to the theory of Hamiltonian PDEs.

The simplest example of a Hamiltonian PDE is the celebrated KdV equation
\begin{equation}\label{intro:eq5}
\frac{du}{dt}=3uu'+cu'''\,,\quad c\in\mb F
\,,
\end{equation}
the only one studied in \cite{ZF71}.
It is Hamiltonian for the algebra of differential polynomials in one variable $\mc V=\mb F[u,u',u'',\dots]$,
the local Hamiltonian functional $\tint h=\frac12\tint u^2$, and the $\lambda$-bracket, defined
on the generator $u$ by
\begin{equation}\label{intro:eq6}
\{u_\lambda u\}=(\partial+2\lambda)u+c\lambda^3\,,
\end{equation}
and uniquely extended to $\mc V$ by the PVA axioms,
which corresponds to the Poisson structure
\begin{equation}\label{intro:eq7}
H(\partial)=u'+2u\partial+c\partial^3\,.
\end{equation}
In fact, equation \eqref{intro:eq5} is Hamiltonian also for another choice of the local Hamiltonian functional and the
Poisson structure: $\tint h_1=\frac12\tint(u^3+cuu'')$, $H_1(\partial)=\partial$, which makes the KdV
equation a bi-Hamiltonian PDE.

For an arbitrary evolution PDE $\frac{d u}{dt}=P$, $P\in\mc V$, an \emph{integral of motion} is a local functional $\tint f\in\mc V/\partial\mc V$
which is conserved by virtue of this equation.
In particular, for the Hamiltonian equation \eqref{intro:eq1} this is equivalent to the property that
$\tint h$ and $\tint f$ are in \emph{involution} (see \eqref{intro:eq2}):
$$
\{\tint h,\tint f\}=0
\,.
$$
A bi-Hamiltonian PDE is called \emph{integrable} if $\tint h$ is contained in an infinite-dimensional abelian subalgebra of
$\mc V/\partial\mc V$
with the Lie algebra bracket \eqref{intro:eq2} for both PVA $\lambda$-brackets on $\mc V$.
Given such a subalgebra of $\mc V/\partial\mc V$, choosing a basis
$\{\tint h_j\}_{j\in J}$, we obtain an integrable hierarchy of compatible bi-Hamiltonian PDEs:
\begin{equation}\label{intro:eq8}
\frac{du}{dt_j}=\{{h_j}_\lambda u\}|_{\lambda=0}\,,\quad j\in J\,.
\end{equation}
It is easy to see from the axioms of a PVA that ``these flows commute'',
meaning that $\frac{d}{dt_i}\frac{d}{dt_j}=\frac{d}{dt_j}\frac{d}{dt_i}$.

We call a PVA $\mc V$ with two compatible $\lambda$-brackets \emph{integrable} if
$\mc V/\partial \mc V$ with the two Lie algebra brackets \eqref{intro:eq2} contains an
infinite-dimensional abelian subalgebra with respect to both of them.

It was the paper \cite{GGKM67} of Gardner, Green, Kruskal and Miura, where they proved that the KdV equation \eqref{intro:eq5}
admits infinitely many linearly independent integrals of motion, which initiated the whole theory of integrable PDEs.
It follows from \cite{ZF71} that these integrals are in involution with respect to the bracket \eqref{intro:eq4}, where
$H(\partial)$ is given by \eqref{intro:eq7}. This leads to the whole KdV hierarchy as in \eqref{intro:eq8}.

In the paper \cite{GD76}, Gelfand and Dickey, using the idea of Lax operator \cite{Lax68}, constructed,
for each integer $N\geq2$, an integrable hierarchy of PDEs, called the $N$th KdV hierarchy,
and in \cite{GD78} they showed that these hierarchies are Hamiltonian. (The corresponding bracket was conjectured earlier by
Adler \cite{Adl79}.) Their $2$nd KdV hierarchy coincides with the classical KdV hierarchy.

In the paper \cite{Mag78} Magri proposed a simple algorithm, called nowadays the Lenard-Magri scheme,
which allows one to prove that integrals of motion of a Hamiltonian PDE \eqref{intro:eq3} are in involution,
provided that the same equation can be written using a different Poisson structure $H_1(\partial)$ in place of $H(\partial)$,
and a different local funtional $\tint h_1$ in place of $\tint h$.
In such a case
one obtains a \emph{bi-Hamiltonian hierarchy} of PDEs. See \cite{BDSK09} for details.

In their seminal paper \cite{DS85} Drinfeld and Sokolov constructed the Lie algebra of local functionals obtained in \cite{GD78}
for the $N$th KdV hierarchy via classical Hamiltonian reduction for the affine Kac-Moody algebra $\widehat{\mf {sl}}_N$.
This \emph{Drinfeld-Sokolov reduction} was developed by them for an affine Kac-Moody algebra $\widehat{\mf g}$, attached to an arbitrary
simple Lie algebra $\mf g$, which led to the construction of an integrable bi-Hamiltonian hierarchy of PDEs,
called the \emph{Drinfeld-Sokolov hierarchy}, attached to an arbitrary affine Kac-Moody algebra.

The Drinfeld-Sokolov reduction is based on a principal nilpotent element $f$ of the simple Lie algebra $\mf g$. It was
extended in \cite{FORTW92} to a certain class of nilpotent elements $f$ of a simple Lie algebra $\mf g$. The construction
of a Drinfeld-Sokolov hierarchy in full generality turned out much more difficult, however.

The construction of the original Drinfled-Sokolov hierarchy in \cite{DS85} is based on Kostant's theorem that
\emph{cyclic elements} $f+E\in\mf g$ (see Definition \ref{20200518:def1}(a)), attached to a principal nilpotent element $f$,
are semisimple. In a series of papers in the early 90's this construction was extended to other nilpotent elements admitting a semisimple cyclic element
(see \cite{DSKV13} for references).

The theory of integrable Hamiltonian hierarchies of PDEs has been naturally related to the theory of PVAs in \cite{BDSK09}.
An important class of PVAs, called the \emph{classical affine $W$-algebras} and denoted by $\mc W(\mf g,f)$,
where $\mf g$ is a simple Lie algebra and $f$ is its nilpotent element,
was considered
in \cite{DSK06} as the quasiclassical limit of quantum affine $W$-algebras. It was then shown in \cite{DSKV13} that the PVA
$\mc W(\mf g,f)$ can be obtained by a classical Hamiltonian reduction in the framework of PVA theory, analogous to the Drinfeld-Sokolov reduction.
In the same paper it was proved that the Lie algebras
of local functionals, constructed in \cite{DS85} and in \cite{FORTW92}, coincide with the Lie algebras $\mc W(\mf g,f)/\partial\mc W(\mf g,f)$ with the bracket \eqref{intro:eq2}.

Furthermore, the Drinfeld-Sokolov hierarchies and their generalizations have been constructed in \cite{DSKV13}
for the $W$-algebras $\mc W(\mf g,f)$, using PVA techniques,
provided that $f\in\mf g$ is a nilpotent element
of \emph{semisimple type}. The latter means that there exists a semisimple cyclic element $f+E$ attached to $f$. This
establishes integrability of the PVA $\mc W(\mf g,f)$ for $f$ of semisimple type.

Unfortunately, the classification of cyclic elements in simple Lie algebras, obtained in \cite{EKV13}, shows that there are
very few semisimple type nilpotent elements in classical Lie algebras, and only about half of the nilpotent elements in exceptional
simple Lie algebras are of semisimple type.

The first basic idea of the present paper is that the Drinfeld-Sokolov method for constructing integrals of motion in \cite{DS85},
extended to the case of nilpotents $f$ of semisimple type in \cite{DSKV13},
generalizes, after a simple modification, to all nilpotents $f$ admitting a non-nilpotent cyclic element $f+E$
(see Theorems \ref{20200518:thm4} and \ref{final}).

The only nilpotent elements $f$ which are left out from the above generalization are those of \emph{nilpotent type},
i.e. such that all the cyclic elements $f+E$ are nilpotent.
According to \cite{EKV13}, there are altogether 15 conjugacy classes of nilpotent
elements of nilpotent type in all exceptional Lie algebras, and, among classical Lie algebras,
they exist only in $\mf{so}_n$ with $n\geq7$, and correspond to partitions $(p_1>p_2=p_1-1\geq\dots)$, where $p_1$ is odd.


In order to treat the nilpotent elements $f$ of nilpotent type, we use the
idea of \cite{DSKV13}, another version of \cite{FGMS95}, that the
Drinfeld-Sokolov method works also for those $f$ which admit a semisimple
\emph{quasi}-cyclic element $f+E$ (see Definition
\ref{20200518:def1} (b)).
In the present
paper we show that this is also the case when $f$ admits a non-nilpotent
quasi-cyclic element (Theorems \ref{20200518:thm5} and \ref{final}).

This establishes integrability of  classical affine $W$-algebras $\mc W(\mf g,f)$ for all classical simple
Lie algebras $\mf g$ and all their nilpotent elements $f$, and for all the exceptional simple Lie algebras $\mf g$
and their nilpotent elements, except, possibly, the following types (in the notation of \cite{CMG93}):
$4A_1$, $2A_2+2A_1$, $2A_3$, $A_4+A_3$ and $A_7$ in $E_8$; $\tilde A_2+A_1$ in $F_4$; $\tilde A_1$ in $G_2$.

\smallskip

The contents of the paper are as follows.

In Section \ref{sec:cyclic} we define the notion of cyclic and quasi-cyclic elements (see Definitions \ref{20200518:def1} (a) and (b))
and discuss their properties.
We show that non-nilpotent cyclic and quasi-cyclic elements
give rise to \emph{integrable triples} in $\mf g$, see Definition \ref{20200518:def3}.
In Section \ref{sec:2}
we recall the notions of a PVA and integrable Hamiltonian PDE, and discuss the Lenard-Magri scheme of integrability.
In Section \ref{sec:3} we recall the construction of classical affine $W$-algebras and
show that any integrable triple
gives rise to an integrable generalized Drinfeld-Sokolov hierarchy (Theorem \ref{final}). This theorem implies integrability of all
classical affine $W$-algebras, associated to classical Lie algebras,
and all classical $W$-algebras  $\mc W(\mf g,f)$, associated to exceptional Lie algebras $\mf g$, except for the seven nilpotents
$f$ mentioned above.

\smallskip
Throughout the paper the base field $\mb F$ is an algebraically closed field of characteristic zero.

\smallskip

\subsubsection*{Acknowledgments}
The first author was partially supported by the national PRIN grant ``Moduli and Lie theory'', and the University grant n.1470755.
The second author was partially supported by the grant FR-18-10849 of Shota Rustaveli
National Science Foundation of Georgia.
The third author was partially supported by the Bert and Ann Kostant
fund.

\section{Cyclic and quasi-cyclic elements. Integrable triples}\label{sec:cyclic}

\subsection{Setup}\label{sec:cyclic.1}

Let $\mf g$ be a reductive finite-dimensional Lie algebra and let $f\in\mf g$ be a non-zero nilpotent element.
Recall that by the Jacobson-Morozov Theorem
\cite[Thm.3.3.1]{CMG93}, any non-zero nilpotent element $f$ is part of an
$\mf{sl}_2$-triple
$\mf s=\{e,h,f\}$ in $\mf g$,
and by Kostant's Theorem \cite[Thm.3.4.10]{CMG93}
all $\mf{sl}_2$-triples containing $f$ are conjugate by the centralizer of $f$ in $G$,
the adjoint group for $\mf g$. It follows that all the constructions of this paper depend only on the $G$-orbit of $f$,
and not on the chosen $\mf{sl}_2$-triple.

We have the $\ad \frac h2$-eigenspace decomposition
\begin{equation}\label{eq:dec}
\mf g=\bigoplus_{k\in \frac{1}{2}\mb Z}\mf g_{k}
\,,
\qquad
\mf g_k=\big\{a\in\mf g\mid [h,a]=2ka\big\}
\,.
\end{equation}
The largest $d\in\frac{1}{2}\mb Z$ such that $\mf g_d\neq0$ is called the \emph{depth} of $f$.
Note that $\dim\mf g_{k}=\dim\mf g_{-k}$ by $\mf{sl}_2$-representation theory, and
$d\geq1$ since $f\in\mf g_{-1}$.
For $j\in\frac12\mb Z$, we shall use the notation $\mf g_{>j}:=\bigoplus_{k>j}\mf g_k$,
and similarly for $\mf g_{\geq j}$, $\mf g_{<j}$ and $\mf g_{\leq j}$.

\begin{remark}\label{rem:depth}
The depth $d$ of $f$ is easy to compute by knowing
the Dynkin characteristic of $f$, defined as follows.
Choose a Cartan subalgebra $\mf h$ of $\mf g$ contained in $\mf g_0$
and choose a set of simple roots $\alpha_1,\dots,\alpha_r$ of $\mf g$ such that $\alpha_i(h)\geq0$, for
all $i=1,\dots,r$.
If $\mf g$ is simple and $\theta$ is the corresponding highest root, then $d=\frac12\theta(h)$.
More concretely, if $\theta=\sum_{i=1}^ra_{i}\alpha_i$, then
$$
d=\frac12\sum_{i=1}^ra_{i}\alpha_i(h)\,.
$$
In the general case of a reductive Lie algebra $\mf g$ the depth of $f$
is equal to the maximum of the depths over all simple components of $\mf g$.
Recall from \cite[\textsection3.5]{CMG93} that $\frac{1}{2}\alpha_i(h)$ can have only the values $0,\frac12$ and $1$;
the collection of these numbers is called the Dynkin characteristic of $f$. Traditionally, the Dynkin characteristic is the collection of integers $\alpha_i(h)$, however, in the theory of $W$-algebras it is more natural to consider the halves of these integers.
\end{remark}

Let $(\cdot\,|\,\cdot)$ be a non-degenerate symmetric invariant bilinear form on $\mf g$.
The subspace $\mf g_{\frac12}$ carries a skew-symmetric bilinear form $\omega$, defined by
\begin{equation}\label{eq:skewform}
\omega(a,b)=(f|[a,b])\,,\quad a,b\in\mf g_{\frac12}\,.
\end{equation}
It is non-degenerate since $\ad f:\mf g_{\frac12}\to\mf g_{-\frac12}$ is an isomorphism
(by $\mf{sl}_2$-representation theory).
\begin{definition}\phantomsection\label{20200518:def1}
\begin{enumerate}[(a)]
\item An element of $\mf g$ of the form
  \begin{equation}\label{20200518:defcyclic}
  f+E\,,\quad \text{where } E\in\mf g_d\backslash\{0\},
  \end{equation}
is called a \emph{cyclic} element attached to the nilpotent element $f$.
\item An element of $\mf g$ of the form
\begin{equation}\label{20200518:defquasicyc}
f+E\,,\quad \text{where } E\in\mf g_{d-\frac12}\backslash\{0\},
\end{equation}
is called a \emph{quasi-cyclic} element attached to the nilpotent element $f$,
if the centralizer of $E$ in $\mf g_{\frac12}$
is coisotropic with respect to the skew-symmetric form $\omega$
(i.e., its orthocomplement in $\mf g_{\frac12}$ is isotropic).
\end{enumerate}
\end{definition}
The following lemma is obvious.
\begin{lemma}\label{20200518:lem1}
\begin{enumerate}[(a)]
\item The subspace $\mf g_d$ lies in the center of the subalgebra $\mf g_{>0}$.
\item
Let $E\in\mf g_{d-\frac12}$ and let $\mf l^\perp$ be the centralizer of $E$ in $\mf g_{\frac12}$.
Then $E$ lies in the center of the subalgebra $\mf n:=\mf l^\perp\oplus\mf g_{\geq1}$.
\end{enumerate}
\end{lemma}

\subsection{Classification of cyclic elements}\label{sec:cyclic.2}
The classification for reductive $\mf g$ easily reduces to the case when $\mf g$ is simple \cite{EKV13}, which we will assume in this subsection. We shall often use the well-known fact that an element of a reductive subalgebra in a reductive Lie algebra $\mf g$
is semisimple (resp. nilpotent) if and only if it is semisimple (resp. nilpotent) in $\mf g$.
\begin{definition}[\cite{EKV13}]\label{20200518:def2}
The nilpotent element $f$ is called of \emph{semisimple type} if there exists a cyclic element attached to $f$
which is semisimple,
and it is called of \emph{nilpotent type} if all cyclic elements attached to $f$ are nilpotent.
\end{definition}
Let $\mf z(\mf s)$ (resp. $Z(\mf s)$) be the centralizer of the $\mf{sl}_2$-triple $\mf s$ in $\mf g$ (resp. in the adjoint group $G$).
More generally, for a subalgebra $\mf q$ of $\mf g$ we denote by $\mf z(\mf q)$ its centralizer in $\mf g$.
Note that if a subalgebra $\mf q$ of $\mf g$ is normalized by $\mf s$, then the $\ad x$-grading of
$\mf g$ induces that of $\mf q$ by letting $\mf q_j=\mf q \cap\mf g_j$.
%
\begin{theorem}[\cite{EKV13}]\phantomsection\label{20200518:thm1}
\begin{enumerate}[(a)]
\item
The nilpotent element $f$ is of nilpotent type if and only if the depth of $f$ is not an integer.
\item
If the cyclic element $f+E$, $E\in\mf g_d$, is semisimple,
then the $Z(\mf s)$-orbit of $E$ in $\mf g_d$ is closed.
If $f$ is of semisimple type, then the set
$\{E\in\mf g_d\mid f+E \text{ is semisimple}\}$ contains a non-empty Zariski open subset.
\item
If $f$ is not of nilpotent type, then
there exists a semisimple subalgebra $\mf q\subset\mf g$
normalized by $\mf s$,
such that $\mf q_d \neq 0$, and an element $f^s\in \mf q_{-1}$ of
semisimple type in $\mf q$, for which the element $f^n:=f-f^s$ lies in $\mf z(\mf q)$. Consequently,
by (b), there exists $E\in \mf q_d$, such that $f^s+E$ is semisimple (in $\mf q$, hence in $\mf g$) and $[f+E,f^n]=0$.
%
\end{enumerate}
\end{theorem}
\begin{proof}
Claims (a) and (b) are proved in \cite[Thm.1.1 \& Prop.2.2]{EKV13}. Claim (c) follows from discussions in Section 4 of \cite{EKV13} for
classical Lie algebras, and Section 5 of \cite{EKV13} for exceptional Lie algebras. 
\end{proof}

\subsection{Quasi-cyclic elements}\label{sec:cyclic.3}
The classification of quasi-cyclic elements in reductive Lie algebras is discussed in \cite{DSJKV20}.
In this subsection we shall discuss only the quasi-cyclic elements in a simple Lie algebra $\mf g$, associated to a nilpotent element
$f$ of nilpotent type.
\begin{theorem}[\cite{EKV13}]\label{20200518:thm2}
\begin{enumerate}[(a)]
\item
There are no nilpotent elements of nilpotent type in $\mf{sl}_n$ and $\mf{sp}_n$.
\item
All nilpotent elements  of nilpotent type in $\mf{so}_n$ correspond to partitions for which the largest part $p_1$ is odd and has
multiplicity $1$, and the next part $p_2$ equals $p_1-1$ and has even multiplicity.
\item
All nilpotent elements $f$ of nilpotent type in exceptional Lie algebras and their depths are listed in Table \ref{table}.
\end{enumerate}
\end{theorem}
\begin{proof}
Part (a) and (b) follow from Section 4 of \cite{EKV13}. The list of Table \ref{table} is Table 1.1 from \cite{EKV13}.
\end{proof}
\begin{theorem}[\cite{DSJKV20}]\label{20200518:thm3}
\begin{enumerate}[(a)]
\item
For all nilpotent elements of nilpotent type in $\mf{so}_n$ there exists a non-nilpotent quasi-cyclic element.
\item
There are no non-nilpotent quasi-cyclic elements for the following nilpotent elements of nilpotent type (see Table \ref{table}):
$4A_1$, $2A_2+2A_1$, $2A_3$, $A_4+A_3$and $A_7$ in $E_8$; $\tilde A_2+A_1$ in $F_4$; $\tilde A_1$ in $G_2$.
\item
For all other nilpotent elements of nilpotent type in exceptional simple Lie algebras there exists a non-nilpotent quasi-cyclic element
(see Table \ref{table}).
\end{enumerate}
\end{theorem}
\begin{proof}
(a) follows from Example \ref{20200518:exa3} and 
  Remark \ref{20200518:rem1} below.
  (b) and (c) follow from the discussion preceding Table \ref{table}, Example \ref{20200612:exa4} and  Remark \ref{20200518:rem1} below.
\end{proof}  
\begin{example}\label{20200518:exa1}
Let $\mf g=G_2$ and let $f$ be the nilpotent element denoted by $\tilde A_1$
as in the last row in Table \ref{table}.
This is a nilpotent element of nilpotent type,
so all cyclic elements are nilpotent. We claim that there are no quasi-cyclic elements attached to $f$.
Recall that the set of positive roots for $\mf g$ is $R_+=\{\alpha,\beta,\alpha+\beta,\alpha+2\beta,\alpha+3\beta,2\alpha+3\beta\}$, where $\alpha$ and $\beta$ are simple roots and $\beta$ is a short root.
Choose root vectors $e_{\gamma}$, $\gamma\in R=R_+\cup(-R_+)$.
%
%
Then, for the grading \eqref{eq:dec},
we have:
$\deg e_{\alpha}=0$,
$\deg e_\beta=\frac12$ (see Table \ref{table}),
so that, for this grading,
$$
\mf g_{\frac12}=\mb Fe_{\beta}\oplus\mb Fe_{\alpha+\beta}\,,\quad
\mf g_{1}=\mb Fe_{\alpha+2\beta}\,,\quad\mf g_{\frac32}=\mb Fe_{\alpha+3\beta}\oplus\mb Fe_{2\alpha+3\beta}
\,.
$$
The centralizer of $e=e_{\alpha+2\beta}$ in $\mf g_{\frac12}$ is zero,
hence it is not coisotropic in $\mf g_{\frac12}$.
Note that $\mf g_{-1}=\mb Fe_{-\alpha-2\beta}$,
so that $f=e_{-\alpha-2\beta}$ is a short root vector.
\end{example}
\begin{example}\label{20200518:exa2}
Let $\mf g$ be a simple Lie algebra, different from $\mf{sp}_n$, and let $f=e_{-\theta}$, the lowest root vector.
In this case, the depth is $d=1$,
and there exists a unique $E\in\mf g_{\frac12}$, up to a non-zero constant factor and action of $Z(\mf s)$,
such that $f+E$ is a semisimple quasi-cyclic element, see \cite[Prop.8.4]{DSKV14}.
\end{example}

In Table \ref{table} we list the Dynkin characteristics of all nilpotent elements $f$ of nilpotent type in exceptional simple
Lie algebras $\mf g$ (the notation is the same as in \cite{CMG93}). In the third column we list the depth $d$ of $f$.
In the fourth column we list the image in $\End\mf g_{d-\frac12}$
of the action of the Lie algebra $\mf z(\mf s)$ on $\mf g_{d-\frac12}$,
and in the fifth the rank of this action, defined as $\dim\left(\mf g_{d-\frac12}\git Z(\mf s)\right)$.
Here for a linear reductive group $G|V$ we use the standard notation
$V\git G=\Spec\mb F[V]^G$.
Finally, the last column says whether there exists a quasi-cyclic element attached to $f$ which is semisimple or non-nilpotent.
It follows from Table \ref{table} that there are no quasi-cyclic elements attached to $f$ if and only if $\dim \mf g_{d-\frac12}=1$.

By $\st_{\mf{a}}$ we denote the standard representation of the Lie algebra $\mf a$
(it is $26$-dimensional for $\mf a=F_4$
and $7$-dimensional for $\mf a=G_2$);
$\bm 1$ stands for the trivial  1-dimensional representation;
$\oplus$ stands for the direct sum of linear Lie algebras:
$\mf (a_1\oplus\mf a_2)|(V_1\oplus V_2)$.
\begin{table}[h]
%
\begin{tabular}{c|c|c|c|c|c}
$\mathfrak g$&nilpotent $f$&$d$&$\mf z(\mathfrak s)|{\mathfrak g}_{d-\frac12}$&rank&$f+E$, $E\in\mathfrak g_{d-\frac12}$\bigstrut[b]\\
\hline
$E_6$&$3A_1$\hfill\EDynkin{0,0,0,\frac12,0,0}{1}&$\frac32$&$\ad_{\mathfrak{sl}_3}\oplus{\bf1}$&$3$&$\exists$ semisimple\bigstrut[t]\\
&$2A_2+A_1$\hfill\EDynkin{0,\frac12,0,\frac12,0,\frac12}{1}&$\frac52$&$\st_{\mathfrak{so}_3}\oplus{\bf1}$&$2$&$\exists$ non-nilpotent\\
\hline
$E_7$&$3A_1'$\hfill\EDynkin{0,0,\frac12,0,0,0,0}{1}&$\frac32$&$\bigwedge^2\st_{\mathfrak{sp}_6}$&$3$&$\exists$ semisimple\bigstrut[t]\\
&$4A_1$\hfill\EDynkin{\frac12,0,0,0,0,0,\frac12}{1}&$\frac32$&$\bigwedge^2\st_{\mathfrak{sp}_6}\oplus{\bf1}$&$4$&$\exists$ semisimple\\
&$2A_2+A_1$\hfill\EDynkin{0,0,\frac12,0,0,\frac12,0}{1}&$\frac52$&$\st_{\mathfrak{so}_3}\oplus\st_{\mathfrak{so}_3}$&$2$&$\exists$ non-nilpotent\\
\hline
$E_8$&$3A_1$\hfill\EDynkin{0,0,0,0,0,0,\frac12,0}{1}&$\frac32$&${\st_{F_4}}\oplus{\bf1}$&$3$&$\exists$ semisimple\bigstrut[t]\\
&$4A_1$\hfill\EDynkin{\frac12,0,0,0,0,0,0,0}{1}&$\frac32$&$\bigwedge^2\st_{\mathfrak{sp}_8}$&$4$&nilpotent only\\
&$2A_2+A_1$\hfill\EDynkin{0,\frac12,0,0,0,0,\frac12,0}{1}&$\frac52$&${\st_{G_2}}\oplus\st_{\mathfrak{so}_3}$&$2$&$\exists$ non-nilpotent\\
&$2A_2+2A_1$\hfill\EDynkin{0,0,0,0,\frac12,0,0,0}{1}&$\frac52$&$\ad_{\mathfrak{so}_5}$&$2$&nilpotent only\\
&$2A_3$\hfill\EDynkin{0,\frac12,0,0,\frac12,0,0,0}{1}&$\frac72$&$\st_{\mathfrak{so}_5}\oplus{\bf1}$&$2$&nilpotent only\\
&$A_4+A_3$\hfill\EDynkin{0,0,0,\frac12,0,0,\frac12,0}{1}&$\frac92$&$\st_{\mathfrak{so}_3}$&$1$&nilpotent only\\
&$A_7$\hfill\EDynkin{0,\frac12,0,\frac12,0,\frac12,\frac12,0}{1}&$\frac{15}2$&${\bf1}$&$1$&never quasi-cyclic\\
\hline
$F_4$&$A_1+\tilde A_1$\hfill\FDynkin{0,0,0,\frac12}&$\frac32$&$S^2\st_{\mathfrak{so}_3}$&$3$&$\exists$ semisimple\bigstrut[t]\\[1em]
&$\tilde A_2+ A_1$\hfill\FDynkin{\frac12,0,0,\frac12}&$\frac52$&$\st_{\mathfrak{so}_3}$&$1$&nilpotent only\bigstrut[b]\\
\hline
$G_2$&$\tilde A_1$\hfill\GDynkin{\frac12,0}&$\frac32$&${\bf1}$&$1$&never quasi-cyclic\bigstrut[t]
\end{tabular}
\caption{Quasi-cyclic elements attached to nilpotent elements of nilpotent type in exceptional simple Lie algebras}\label{table}
\end{table}

The results of Table \ref{table} are obtained as follows (see \cite{DSJKV20} for
details). Denote $\mf g_{d-\frac12}$ by $\mf m$.
First, we compute $\mf z(\mf s)|\mf m$ using the SLA package in the GAP
computer algebra system. We note that, by Table \ref{table},  $Z(\mf s)| \mf m$ is a polar representation with
a Cartan subspace $\mf m_0$ (see \cite{DK85} for the definitions). In particular, any closed orbit of $Z(\mf s)$ in $\mf m$
intersects $\mf m_0$ non-trivially \cite{DK85}.
It follows that if $\mf m_0$ contains no elements $E_0$, such that $f+E_0$ is
quasi-cyclic, then $\mf m$ contains no elements $E$, such that $f+E$ is non-nilpotent
quasi-cyclic. Indeed, in the contrary case the orbit
of minimal dimension in $Z(\mf s)E$ is closed and non-zero, hence there exists
$E_0\in \mf m_0$, which lies in this orbit, such that $f+E_0$ is quasi-cyclic.
This allows us to restrict consideration of quasi-cyclicity of $f+E$ to $E\in \mf m_0$.
Then, one uses again SLA as
well as Mathematica to study non-nilpotence or semisimplicity of
quasi-cyclic elements $f + E$ with $E\in\mf m_0$ through minimal polynomials of
their matrices in a faithful representation of a semisimple subalgebra
containing them.

\subsection{Integrable triples}\label{sec:cyclisc.4}
Let $\mf g$ be a finite-dimensional Lie algebra,
let $(\cdot\,|\,\cdot)$ be a non-degenerate symmetric invariant bilinear form on $\mf g$,
let $f\in\mf g$ be a non-zero nilpotent element contained in an $\mf{sl}_2$-triple $\mf s\subset\mf g$,
and consider the corresponding
$\frac{1}{2}\mb Z$-grading \eqref{eq:dec}.
%
%
%
\begin{definition}\label{20200518:def3}
  An \emph{integrable triple} associated to $f$ is $(f_1,f_2,E)$, where $f_1,f_2\in \mf g_{-1}$ and
  $E\in\mf g_{\geq\frac12}$ is a non-zero homogeneous element, such that the following three properties hold:
\begin{enumerate}[(i)]
\item $f=f_1+f_2$ and $[f_1,f_2]=0$,
\item $[E, \mf g_{\geq 1}]=0$ and the centralizer of $E$ in $\mf g_{\frac{1}{2} }$ is
coisotropic with respect to the bilinear form \eqref{eq:skewform}.
\item $f_1+E$ is semisimple and $[f_2, E]=0$.
\end{enumerate}  
In this case $E$ is called an \emph{integrable element} for $f$.
\end{definition}
Note that for an integrable triple $(f_1,f_2,E)$ the decomposition $f+E=(f_1+E)+ f_2$ is  a Jordan decomposition
of $f+E$, and that  $E$ is a central element of the subalgebra $\mf n:=\mf l^\perp\oplus\mf g_{\geq1}$, 
%


\begin{theorem}\label{20200518:thm4}
Let $f$ be a nilpotent element of integer depth $d$ of a reductive Lie algebra $\mf g$.
\begin{enumerate}[a)]
\item If f is of semisimple type, then there exists $E\in \mf g_d$, such that
$(f,0,E)$ is an integrable triple in $\mf g$.
\item If $f^s$, $f^n$ and $E\in \mf g_d$ are the elements, constructed in Theorem \ref{20200518:thm1}(c), then $(f^s, f^n, E)$ is an integrable triple in $\mf g$.
\end{enumerate}
%
\end{theorem}
\begin{proof}
It follows immediately from Theorem \ref{20200518:thm1}(c).
\end{proof}

\begin{example}\label{20200518:exa3}
Let $\mf g=\mf{s0}_N$, $N\geq7$, and let $f\in\mf g$ be a nilpotent element of nilpotent type.
Recall from \cite{EKV13}
that this nilpotent element corresponds to a partition $\underline p=(p_1^{r_1},p_2^{r_2},\dots,p_n^{r_n})$ of $N$
(namely, $N=r_1p_1+\dots+r_np_n$, where
$p_1>p_2>\dots>p_n>0$, $r_i\geq1$, for $1\leq i\leq n$), and $p_1=p+1$ occurs with multiplicity $r_1=1$, $p_2=p$, and $p$ is even (this implies that $r_2$ is also even).
Then, there exists an integrable triple
$(f_1,f_2,E)$ associated to $f$, with $E\in\mf g_{d-\frac12}$, constructed as follows.

Let $V=\bigoplus_{\alpha\in I}\mb Fe_\alpha\cong \mb F^N$ be the $N$-dimensional vector space with basis $e_\alpha$,
$\alpha\in I=\{(a,i,j)|1\leq a\leq n, 1\leq i\leq r_a,1\leq j\leq p_a\}$.
Consider the following involution on the set $I$
\begin{equation}\label{20200424:eq1}
(a,i,j)'=(a,r_a+1-i,p_a+1-j)
\,,
\end{equation}
and also define
\begin{equation}\label{20200424:eq2}
\epsilon_{(a,i,j)}=
\left\{
\begin{array}{ll}
(-1)^{(i-1)p_a+j}\,,&1\leq i\leq \lceil\frac{r_a}{2}\rceil\,,
\\
-(-1)^{(i-1)p_a+j+r_a}\,,&\lceil\frac{r_a}{2}\rceil+1\leq i\leq r_a\,,
\end{array}
\right.
\end{equation}
where $\lceil\cdot\rceil$ denotes the ceiling function.

Define a non-degenerate bilinear form on $V$ on basis elements as follows:
\begin{equation}\label{20200422:form}
\langle e_\alpha|e_\beta\rangle=-\delta_{\alpha,\beta'}\epsilon_{\alpha}
\,,
\qquad \alpha,\beta\in I\,.
\end{equation}
Let $A^\dagger$ denote the adjoint of $A\in\End V$ with respect to \eqref{20200422:form}.
Explicitly, in terms of elementary matrices, it is given by:
\begin{equation}\label{eq:dagger}
(E_{\alpha\beta})^\dagger=\epsilon_\alpha\epsilon_\beta E_{\beta'\alpha'}\,.
\end{equation}
It follows from equations \eqref{20200424:eq1} and \eqref{20200424:eq2} that $\epsilon_\alpha=\epsilon_{\alpha'}$, for every
$\alpha\in I$. Hence, the bilinear form \eqref{20200422:form} is symmetric and  $\mf g=\{A\in\End V\mid A^\dagger=-A\}\simeq\mf{so}_N$.

Consider the element
$$
f=\sum_{(a,i,j)\in I,j\neq p_a}E_{(a,i,j+1),(a,i,j)}\in\mf g
\,.
$$
It is a nilpotent element in the Jordan form corresponding to the partition $\underline p$. We include $f$ in the $\mf{sl}_2$-triple
$\{e,h,f\}\subset\mf g$, where
\begin{equation}\label{sl2triple}
h=\sum_{(a,i,j)\in I}(p_a+1-2j)E_{(a,i,j),(a,i,j)}\,,
\quad
e=\sum_{(a,i,j)\in I,j\neq p_a}j(p_a-j)E_{(a,i,j),(a,i,j+1)}
\,.
\end{equation}
For $\alpha,\beta\in I$, we define the following elements of $\mf g$:
$$
F_{\alpha\beta}=E_{\alpha\beta}-\epsilon_\alpha\epsilon_\beta E_{\beta'\alpha'}\,\,\big(=-F_{\alpha\beta}^\dagger\big)
\,.
$$
The following commutation relations hold ($\alpha,\beta,\gamma,\eta\in I$):
\begin{equation}\label{comm:BC}
[F_{\alpha\beta},F_{\gamma\eta}]=\delta_{\gamma\beta}F_{\alpha\eta}
-\delta_{\eta\alpha}F_{\beta\gamma}-\epsilon_\alpha\epsilon_\beta\delta_{\alpha'\gamma}F_{\beta'\eta}
+\epsilon_\alpha\epsilon_\beta\delta_{\eta\beta'}F_{\gamma\alpha'}
\,.
\end{equation}
Using \eqref{comm:BC} and the explicit form of $h\in\mf g$ given in \eqref{sl2triple}, we have that the depth of the
decomposition \eqref{eq:dec} is $d=\frac{2p-1}{2}$ and
\begin{align*}
\mf g_{\frac12}=&\Span_{\mb F}\{F_{(1,1,j),(2,i,j)}\mid 1\leq i\leq r_2,1\leq j\leq p\}
\\
&\oplus\Span_{\mb F}\{F_{(a,i,i'),(b,j,j')}\mid a,b\geq2\,,\frac{p_a-p_b}{2}-i'+j'=\frac12\}=U\oplus W\,.
\end{align*}
Consider the element $E=F_{(1,1,1)(1,1,p)}\in\mf g_{p-1}=\mf g_{d-\frac12}$. Clearly, $[E,\mf g_{\geq1}]=0$.
Moreover, by the commutation relations \eqref{comm:BC} we also have that $[E,W]=0$.
Let
$$
A=\sum_{i=1}^{r_1}\sum_{j=1}^pa_{ij}F_{(1,1,j),(2,i,j)}\in U
\,.
$$
Using the commutation relations \eqref{comm:BC} we have
$$
[E,A]=\sum_{i=1}^{r_2}a_{ip}F_{(1,1,p),(2,i,p)}
\,.
$$
This implies that the centralizer of $E$ in $\mf g_{\frac12}$ is
\begin{equation}\label{20200619:eq3}
\mf z(E)\cap\mf g_{\frac12}=\Span_{\mb F}\{F_{(1,1,j),(2,i,j)}\mid 1\leq i\leq r_2,1\leq j\leq p-1\}\oplus W
\,.
\end{equation}
We claim that this is a coisotropic subspace of $\mf g_{\frac12}$, namely $\mf z(E)\cap\mf g_{\frac12}=\mf l^{\perp}$
for some isotropic subspace $\mf l\subset\mf g_{\frac12}$.
In fact, consider the following subspace of $\mf g_{\frac12}$:
$$
\mf l=\Span_{\mb F}\{F_{(1,1,1),(2,i,1)}\mid 1\leq i\leq r_2\}
\,.
$$
Since, by the commutation relations \eqref{comm:BC} and the fact that $F_{(1,1,1),(1,1,p+1)}=0$, we have that $[\mf l,\mf l]=0$, it follows that
$\mf l$ is an isotropic subspace with respect to the skew-symmetric bilinear form given in \eqref{eq:skewform} (we assume that
the bilinear form on $\mf g$ in Section \ref{sec:cyclic.1} is the trace form). Moreover, from \eqref{comm:BC} we also have that
$[\mf l,W]=0$. Hence, $W\subset\mf l^{\perp}$. Using again the commutation relations \eqref{comm:BC} we get
($1\leq i,j\leq r_2,1\leq k\leq p$)
$$
[F_{(1,1,1),(2,i,1)},F_{(1,1,k),(2,j,k)}]=\epsilon_{(2,i,1)}\delta_{i+j,r_2+1}\delta_{k,p}F_{(1,1,k),(1,1,p+1)}
\,.
$$
It follows that
$$
\omega(F_{(1,1,1),(2,i,1)},F_{(1,1,k),(2,j,k)})=(f|[F_{(1,1,1),(2,i,1)},F_{(1,1,k),(2,j,k)}])\neq0
$$
if and only if $i+j=r_2+1$ and $k=p$. Hence, $\mf z(E)\cap\mf g_{\frac12}=\mf l^{\perp}$ (cf. equation \eqref{20200619:eq3})
thus showing that $f+E$ is a quasi-cyclic element attached to $f$.

Let $f_1=\sum_{j=1}^pE_{(1,1,j+1),(1,1,j)}\in\mf g$ and $f_2=f-f_1$. It is immediate to check that $[f_1+E,f_2]=0$ and
that $f_2$ lies in $\mf g_{-1}$ (hence, it is nilpotent). Moreover, the minimal polynomial of $f_1+E$
is $p(\lambda)=\lambda(\lambda^p-2)$. Hence, $f_1+E$ is semisimple and $(f_1,f_2,E)$ is an integrable triple associated to $f$.
\end{example}

\begin{example}\label{20200612:exa4}
In simple Lie algebras of exceptional types there are three cases of nilpotent orbits of nilpotent type admitting non-nilpotent quasi-cyclic elements but no semisimple quasi-cyclic elements. Namely, these are nilpotent elements with label $2A_2 +A_1$ in $E_6, E_7$, and $E_8$ (see Table \ref{table}).  For them we find integrable triples computationally, using the SLA package of the GAP system.

We choose the following representatives $f$:

for $E_6$, $f_{\foresix{.1}{.5}101100}+f_{\foresix{.1}{.5}010111}+f_{\foresix{.1}{.5}001111}+f_{\foresix{.1}{.5}111110}+f_{\foresix{.1}{.5}011210}$;

for $E_7$, $f_{\foreseven{.1}{.5}1111110}+f_{\foreseven{.1}{.5}1011111}+f_{\foreseven{.1}{.5}0112110}+f_{\foreseven{.1}{.5}1122100}+f_{\foreseven{.1}{.5}0112211}$;

for $E_8$, $f_{\foreight{.5}11222210}+f_{\foreight{.5}11122211}+f_{\foreight{.5}01122221}+f_{\foreight{.5}12232110}+f_{\foreight{.5}11232111}$.

In all three cases, the zero weight space of the representation of $\mf{z}(\mf{s})$ on $\mf{g}_{d-1/2}$ is a Cartan subspace, spanned by two root vectors $b_1,b_2$, where

for $E_6$, $b_1=e_{\foresix{.1}{.5}111211}$ and $b_2=e_{\foresix{.1}{.5}112221}$;

for $E_7$, $b_1=e_{\foreseven{.1}{.5}1123221}$ and $b_2=e_{\foreseven{.1}{.5}1223321}$;

for $E_8$, $b_1=e_{\foreight{.5}23354321}$ and $b_2=e_{\foreight{.5}22454321}$.

In all three cases, an element $x_1b_1+x_2b_2$ of this subspace has coisotropic centralizer in $\mf{g}_{1/2}$ if and only if $x_2=-x_1$.
Take $E=c(b_1-b_2)$, then in all three cases for any $c\ne0$ the Jordan decomposition of $f+E$ is $(f^s+E)+f^n$, where

for $E_6$, $f^n=f_{\foresix{.1}{.5}011210}$;

for $E_7$, $f^n=f_{\foreseven{.1}{.5}1122100}$;

for $E_8$, $f^n=f_{\foreight{.5}01122221}$.

Thus, since $f^n\in\mf{g}_{-1}$, we have $f^s=f-f^n\in \mf{g}_{-1}$. Hence, $(f^s,f^n,E)$ is an integrable triple.

Additionally, let us remark that in all three of these cases $f^s$ has label $2A_2$ and $f^n$ has label $A_1$. Note also that the subalgebra generated by $f$ and $E$ is the direct sum of an $\mf{sl}_3$ and a 1-dimensional center spanned by $f^n$. Moreover $f^s$ is a principal nilpotent in this $\mf{sl}_3$.
\end{example}

\begin{theorem}\label{20200518:thm5}
Let $f$ be a nilpotent element of depth $d$ of nilpotent type in a simple Lie algebra $\mf g$,
such that there exists $E\in\mf g_{d-\frac12}$, for which $f+E$ is a non-nilpotent quasi-cyclic element.
\begin{enumerate}[(a)]
\item
If $f+E$ is semisimple, then $(f,0,E)$ is an integrable triple associated to $f$.
\item
For $f$, not covered by (a), there exists an integrable triple.
\end{enumerate}
\end{theorem}
\begin{proof}
  Part (a) is clear. The cases, not covered by (a) are as follows.
  First, it is $\mf g = \mf {so}_N$, which is covered by Example \ref{20200518:exa3}.
  Second, it is $2A_2+A_1$ in all algebras of type $E$, which is covered by Example \ref{20200612:exa4}.
\end{proof}
\begin{remark}\label{20200518:rem1}
Conversely, if $(f_1,f_2,E)$ is an integrable triple for $f$, where $E\in\mf g_d$ (resp. $E\in\mf g_{d-\frac12}$), then $f_1+f_2+E$
is a non-nilpotent cyclic (resp. quasi-cyclic) element. This follows from Definition \ref{20200518:def1}.
\end{remark}
\begin{remark}\label{20200525:rem1}
It is clear from Example \ref{20200518:exa1} that for $\mf g=G_2$ and $f$ a short root vector there are no integrable triples.
\end{remark}
%

\section{Poisson vertex algebras, Hamiltonian equations and integrability}\label{sec:2}

\subsection{PVA and Hamiltonian PDE}\label{sec:2.2}

Recall (see e.g. \cite{BDSK09})
that a \emph{Poisson vertex algebra} (abbreviated PVA) is a commutative associative algebra $\mc V$ with a derivation $\partial$,
endowed with a $\lambda$-\emph{bracket} $\mc V\otimes \mc V\to \mc V[\lambda]$,
$a\otimes b\mapsto\{a_\lambda b\}$,
satisfying the following axioms ($a,b,c\in \mc V$):
\begin{enumerate}[(i)]
\item
sesquilinearity:
$\{\partial a_\lambda b\}=-\lambda\{a_\lambda b\}$,
$\{a_\lambda\partial b\}=(\partial+\lambda)\{a_\lambda b\}$;
\item
skew-symmetry:
$\{a_\lambda b\}=-\{b_{-\lambda-\partial}a\}$
(where $\partial$ acts on the coefficients);
\item
Jacobi identity:
$\{a_\lambda\{b_\mu c\}\}-\{b_\mu\{a_\lambda c\}\}=\{\{a_\lambda b\}_{\lambda+\mu} c\}$;
\item
left Leibniz rule:
$\{a_\lambda bc\}=\{a_\lambda b\}c+b\{a_\lambda c\}$.
\end{enumerate}
As a consequence of skew-symmetry and the left Leibniz rule, we also have the
\begin{enumerate}[(i')]
\setcounter{enumi}{3}
\item
right Leibniz rule:
$\{ab_\lambda c\}=
\{a_{\lambda+\partial} c\}_\rightarrow b
+\{b_{\lambda+\partial} c\}_\rightarrow a$, where $\rightarrow$ means that $\partial$ is moved to the right.
\end{enumerate}
\begin{example}\label{ex:affinePVA}
The most important example for this paper will be the affine PVA $\mc V(\mf g,E)$, where $\mf g$ is a Lie algebra with a
symmetric invariant bilinear form $(\cdot\,|\,\cdot)$
and $E\in\mf g$. It is defined as the differential algebra $\mc V(\mf g)=S(\mb F[\partial]\mf g)$, the algebra of differential polynomials over
the vector space $\mf g$, with the PVA $\lambda$-bracket given by
\begin{equation}\label{eq:lambda-affine}
\{a_\lambda b\}_z=[a,b]+(a|b)\lambda+z(E|[a,b])
\,,\quad a,b\in\mf g,z\in\mb F\,,
\end{equation}
and extended to $\mc V$ by sesquilinearity axioms and the Leibniz rules.
\end{example}

Let $\mc V$ be a PVA. We denote by $\tint:\,\mc V\to\mc V/\partial\mc V$ the canonical quotient map
of vector spaces.
Recall that (see \cite{BDSK09}) $\mc V/\partial\mc V$ carries a well-defined Lie algebra structure given by
$$
\{\tint f,\tint g\}=\tint\{f_\lambda g\}|_{\lambda=0}
\,,
$$
and we have a representation of the Lie algebra $\mc V/\partial\mc V$ on $\mc V$
given by
$$
\{\tint f,g\}=\{f_\lambda g\}|_{\lambda=0}
\,.
$$

A Hamiltonian equation on $\mc V$ associated to a Hamiltonian functional
$\tint h\in\mc V/\partial\mc V$ is the evolution equation
\begin{equation}\label{ham-eq}
\frac{du}{dt}=\{\tint h,u\}\,\,, \,\,\,\, u\in\mc V\,.
\end{equation}
The minimal requirement for integrability is to have an infinite collection
of linearly independent integrals of motion in involution:
$$
\tint h_0=\tint h,\,\tint h_1,\,\tint h_2,\,\dots\,
\,\text{ such that }\,\, \{\tint h_m,\tint h_n\}=0\,\,\text{ for all }\,\, m,n\in\mb Z_{\geq0}
\,.
$$
In this case, we have the \emph{integrable hierarchy} of Hamiltonian equations
\begin{equation}\label{eq:int-hier}
\frac{du}{dt_n}=\{\tint h_n,u\}\,\,, \,\,\,\, n\in\mb Z_{\geq0}
\,.
\end{equation}

\subsection{Bi-Poisson vertex algebras and Lenard-Magri scheme of integrability}\label{sec:3.1}

Let $\{\cdot\,_\lambda\,\cdot\}_0$ and $\{\cdot\,_\lambda\,\cdot\}_\infty$
be two $\lambda$-brackets on the same differential algebra $\mc V$.
We can consider the pencil of $\lambda$-brackets
\begin{equation}\label{20200627:eq1}
\{\cdot\,_\lambda\,\cdot\}_z
=
\{\cdot\,_\lambda\,\cdot\}_0
+
z\{\cdot\,_\lambda\,\cdot\}_\infty
\quad,\qquad z\in\mb F
\,.
\end{equation}
As above, we say that $\mc V$ is a \emph{bi-PVA}
if $\{\cdot\,_\lambda\,\cdot\}_z$ is a PVA $\lambda$-bracket on $\mc V$ for every $z\in\mb F$.
%
%
The affine PVA defined in Example \ref{ex:affinePVA} is in fact a bi-PVA.

Let $\mc V$ be a bi-Poisson vertex algebra with $\lambda$-brackets
$\{\cdot\,_\lambda\,\cdot\}_0$ and $\{\cdot\,_\lambda\,\cdot\}_\infty$.
A \emph{bi-Hamiltonian} equation is an evolution equation
which can be written in Hamiltonian form with respect to both PVA $\lambda$-brackets
and two Hamiltonian functionals $\tint h_0,\tint h_1\in\mc V/\partial\mc V$:
$$
\frac{du}{dt}
=
\{\tint h_0,u\}_0
=
\{\tint h_1,u\}_\infty
\,,\,\, u\in\mc V
\,.
$$

The most common way to prove integrability for a bi-Hamiltonian equation
is to solve the so called \emph{Lenard-Magri
recurrence relation} (see \cite{Mag78}):
\begin{equation}\label{eq:LM}
\{\tint h_n,\tint u\}_0
=
\{\tint h_{n+1},\tint u\}_\infty
\,\,,\,\,\,\,
n\in\mb Z_{\geq0}\,,\, u\in\mc V
\,.
\end{equation}
The Lenard-Magri recurrence relation \eqref{eq:LM} produces local functionals in involution.
In order to prove this claim, one needs the following simple lemma (cf. \cite{Mag78,BDSK09}).
\begin{lemma}\label{20200609:lem1}
Let $\mc U$ be a vector space with two skew-commutative brackets $\{\cdot\,,\,\cdot\}_0$ and $\{\cdot\,,\,\cdot\}_{\infty}$
(not necessarily satisfying Jacobi identity). Let $h_0,h_1,\dots, h_N$, where $N\geq1$, be a sequence of elements of $\mc U$,
satisfying the relation
\begin{equation}\label{eq:LM2}
\{h_n,u\}_0=\{h_{n+1},u\}_{\infty}\,,\quad n=0,\dots,N-1\,,\,u\in\mc U\,.
\end{equation}
Then
\begin{enumerate}[a)]
\item $\{h_m,h_n\}_0=0=\{h_m,h_n\}_{\infty}$, for all $m,n=0,\dots,N$.
\item If $N=\infty$, $\{h_0,\mc U\}_{\infty}=0$ and $\{g_n\}_{n\in\mb Z_{\geq 0}}$ is another sequence satisfying \eqref{eq:LM2},
then $\{h_m,g_n\}_0=0=\{h_m,g_n\}_\infty$, for every $m,n\in\mb Z_{\geq0}$.
\end{enumerate}
\end{lemma}
\begin{proof}
Due to skew-symmetry of the brackets, part a) holds for $m=n$. Without loss of generality we may assume
that $m>n$ and prove the result by induction on $m-n$. Then, by equation \eqref{eq:LM2} for $u=h_m$ and skew-symmetry we have
$\{h_m,h_n\}_0=\{h_m,h_{n+1}\}_\infty=0$ by the inductive assumption.
Similarly, by equation \eqref{eq:LM2} we have $\{h_{m},h_n\}_{\infty}=\{h_{m-1},h_n\}_0=0$. This proves part a).

We prove part b) by induction on $m\in\mb Z_{\geq0}$. For $m=0$ we have,  by our assumption on $h_0$,
$\{h_0,g_n\}_\infty=0$, for every $n\in\mb Z_{\geq 0}$. Furthermore, by our assumption on $h_0$ and skew-symmetry we have
$\{h_0,g_m\}_0=\{h_0,g_{m+1}\}_\infty=0$.
For $m\geq1$ we have, by equation \eqref{eq:LM2} and the inductive assumption,  that $\{h_m,g_n\}_\infty=\{h_{m-1},g_n\}_0=0$, for every $n\in\mb Z_{\geq 0}$. Similarly, using skew-symmetry and the Lenard-Magri recursion \eqref{eq:LM2} for $\{g_n\}_{n\in\mb Z_{\geq 0}}$, we have
$\{h_m,g_n\}_0=\{h_m,g_{n+1}\}_{\infty}=0$, thus concluding the proof of part b).
\end{proof}
As a special case of Lemma \ref{20200609:lem1},
if $\mc V$ is a bi-PVA and $\tint h_0,\tint h_1,\dots\in\mc V/\partial\mc V$
satisfy the Lenard-Magri recurrence \eqref{eq:LM},
then they are in involution:
$$
\{\tint h_m,\tint h_n\}_0
=
\{\tint h_m,\tint h_n\}_\infty
=
0
\,\,\text{ for all }\, m,n\geq0
\,,
$$
namely
$$
\mc A:=\Span\{\tint h_n\}_{n=0}^\infty\subset\mc V/\partial\mc V
$$
is an abelian subalgebra with respect to both Lie algebra brackets $\{\cdot\,,\,\cdot\}_0$ and $\{\cdot\,,\,\cdot\}_\infty$.
In this way,
we get the corresponding hierarchy of bi-Hamiltonian equations
$$
\frac{du}{dt_n}
=
\{\tint h_n,u\}_0
=
\{\tint h_{n+1},u\}_\infty
\,,\,\,
n\in\mb Z_{\geq0},\, u\in\mc V
\,.
$$
If moreover
\begin{equation}\label{eq:casimir}
\{\tint h_{0},\tint u\}_\infty=0
\,\,\text{ for all }\, u\in\mc V
\,,
\end{equation}
and $\tint g_0,\tint g_1,\dots\in\mc V/\partial\mc V$ is any other sequence satisfying \eqref{eq:LM},
then
$$
\widetilde{\mc A}=\Span\{\tint h_n,\,\tint g_n\}_{n=0}^\infty\subset\mc V/\partial\mc V
$$
is also an abelian subalgebra.

\section{Classical affine \texorpdfstring{$W$}{W}-algebras and generalized Drinfeld-Sokolov hierarchies}\label{sec:3}

\subsection{Definition of classical affine $W$-algebras}\label{sec:2.3}

Let $\mf g$ be a reductive Lie algebra as in Section \ref{sec:cyclic.1}, $(\cdot\,|\,\cdot)$ a non-degenerate symmetric invariant bilinear form on $\mf g$,
$f$ its non-zero nilpotent element, and consider the corresponding
$\frac{1}{2}\mb Z$-grading \eqref{eq:dec} and the skew-symmetric non-degenerate bilinear form $\omega$ on $\mf g_{\frac12}$
defined by \eqref{eq:skewform}.
Fix an isotropic (with respect to $\omega$) subspace $\mf l\subset\mf g_{\frac12}$
and denote by $\mf l^\perp=\{a\in\mf g_{\frac12}\mid \omega(a,b)=0 \text{ for all }b\in\mf l\}\subset\mf g_{\frac12}$
its orthogonal complement with respect to $\omega$.
Throughout the paper we consider the following nilpotent subalgebras of $\mf g$:
$$
\mf m = \mf l \oplus \mf g_{\geq1}\subset\mf n = \mf l^\perp\oplus\mf g_{\geq1}\,,
$$
where $\mf g_{\geq1}=\oplus_{k\geq1}\mf g_{k}$.

Fix an element $E\in\mf z(\mf n)$ (the centralizer of $\mf n$ in $\mf g$)
and consider the affine bi-PVA $\mc V(\mf g,E)$ 
from Example \ref{ex:affinePVA}: it is the differential algebra $\mc V(\mf g)=S(\mb F[\partial] \mf g)$,
with its two PVA $\lambda$-brackets $\{\cdot\,_\lambda\,\cdot\}_0$ and $\{\cdot\,_\lambda\,\cdot\}_\infty$
(defined on $a,b\in\mf g$ by $\{a_\lambda b\}_0=[a,b]+(a|b)\lambda$ and
$\{a_\lambda b\}_\infty=(E|[a,b])$, cf. equations
\eqref{eq:lambda-affine} and \eqref{20200627:eq1}).
Consider also the differential algebra ideal $I$, generated by the set
$$
\big\langle m-(f|m) \big\rangle_{m\in\mf m}\,\subset\mc V(\mf g)
\,.
$$
Note that, since $E\in\mf z(\mf n)$, we have $\{a_\lambda w\}_\infty=0$ for all $a\in\mf n$ and $w\in\mc V(\mf g,E)$.
Consider the space
$$
\widetilde{\mc W}
=
\big\{w\in\mc V(\mf g)\,\big|\, \{a_\lambda w\}_0\in I[\lambda]\,,\text{ for every }a\in\mf n\big\}
\,\subset\,
\mc V(\mf g,E)
\,.
$$
\begin{lemma}\label{lem:clw}
$\widetilde{\mc W}\subset\mc V(\mf g,E)$ is a bi-PVA subalgebra of $\mc V(\mf g,E)$
and
$I\subset\widetilde{\mc W}$ is a bi-PVA ideal.
\end{lemma}
\begin{proof}
Straightforward, see e.g. Section 3 of \cite{DSKV13}.
\end{proof}
\begin{definition}\label{def:clw}
The \emph{classical affine} $W$-\emph{algebra} associated to the triple $(\mf g,f,E)$
is the quotient
\begin{equation}\label{eq:W}
\mc W(\mf g,f,E)
=
\widetilde{\mc W}/I
\,,
\end{equation}
which, by Lemma \ref{lem:clw}, has a natural structure of a bi-PVA,
with the induced PVA $\lambda$-brackets $\{\cdot\,_\lambda\,\cdot\}_0$ and $\{\cdot\,_\lambda\,\cdot\}_\infty$.
\end{definition}
\begin{remark}\label{rem:200226}
Recall from \cite{DSKV13} that the classical affine $W$-algebra depends only
on the Lie algebra $\mf g$ and on the nilpotent orbit of $f$:
for different choices of $f$ in its nilpotent orbit,
of the $\mf{sl}_2$-triple $\mf s$ containing $f$, and of the isotropic subspace $\mf l\subset\mf g_{\frac12}$,
we get isomorphic $W$-algebras.
\end{remark}

Let $\mf p\subset\mf g$ be a subspace complementary to $\mf m$ in $\mf g$: $\mf g=\mf m\oplus\mf p$.
We assume that $\mf p$ is compatible with the grading \eqref{eq:dec},
so that $\mf g_{\leq 0}\subset\mf p\subset\mf g_{\leq\frac12}$.
Clearly, we can identify, as differential algebras $\mc V(\mf g)/I\simeq \mc V(\mf p)$.
Hence, the classical $W$-algebra can be viewed as a differential subalgebra of $\mc V(\mf p)$,
the algebra of differential polynomials over $\mf p$:
\begin{equation}\label{eq:W2}
\mc W(\mf g,f,E)
=
\widetilde{\mc W}/I
\subset
\mc V(\mf g)/I
\simeq
\mc V(\mf p)
\,.
\end{equation}

\subsection{Integrable triples and generalized Drinfeld-Sokolov hierarchies}\label{sec:3.2}

In \cite{DSKV13} a generalized Drinfeld-Sokolov hierarchy was constructed, using the Lenard-Magri recurrence relation \eqref{eq:LM}, 
under the assumption that the element $f+E\in\mf g$ is semisimple. In this section we extend this result by constructing
a generalized Drinfeld-Sokolov hierarchy for any integrable triple associated to $f$.

Let $(f_1,f_2,E)$ be an integrable triple associated to $f$ (cf. Definition \ref{20200518:def3}), 
and let $\mf l^\perp$ be the centralizer of $E$ in $\mf g_{\frac12}$,
with $\mf l$ isotropic with respect to \eqref{eq:skewform}.
Setting as above $\mf n=\mf l^\perp\oplus\mf g_{\geq1}$,
we have by definition that $[E,\mf n]=0$.

Let $\mb K = \mb F ((z^{-1}))$ be the field of formal Laurent series in $z^{-1}$ over $\mb F$. Consider the Lie algebra $\mf g ((z^{-1}))
=\mf g\otimes_{\mb F} \mb K$.
Since, by Definition \ref{20200518:def3}, $f_1+E\in\mf g$ is semisimple, 
$f_1+zE\in\mf g((z^{-1}))$ is also semisimple.
Indeed, for a non-zero element $t\in \mb K$ we have 
the Lie algebra automorphism $\varphi_t$ of $\mf g ((z^{-1}))$ acting as $t^i$ on $\mf g_i$.
Hence, if $E\in\mf g_k$, we have $t\varphi_t(f_1+E)=f_1+t^{k+1}E$.
Then, working over the
field extension of $\mb K$, containing $z^{\frac1{k+1}}$,
we conclude that
$f_1+zE$ is semisimple in $\mf g((z^{-1}))$.
We thus have the direct sum decomposition
\begin{equation}\label{eq:dec-gz}
\mf g((z^{-1}))=\mf h\oplus\mf h^\perp
\,,
\end{equation}
where
\begin{equation}\label{eq:dec-gz2}
\mf h
:=
\ker\ad (f_1+zE)
\,\,\text{ and }\,\,
\mf h^\perp
:=
\im\ad(f_1+zE)
\,.
\end{equation}
The notation $\mf h^\perp$ relates to the fact that $\im\ad(f_1+zE)$
is the orthogonal complement of $\ker\ad (f_1+zE)$
with respect to the non-degenerate symmetric invariant
bilinear form $(\cdot\,|\,\cdot)$
on $\mf g((z^{-1}))$, extending the form $(\cdot\,|\,\cdot)$ on $\mf g$ by bilinearity.
(But we will not use this fact.)
%
\begin{lemma}\label{lem:main}
For every element $A(z)\in\mf g((z^{-1}))$, there exist unique $h(z)\in\mf h$ and $U(z)\in\mf h^\perp$ such that
\begin{equation}\label{toprove}
h(z)+[f+zE, U(z)]=A(z)
\,.
\end{equation}
\end{lemma}
\begin{proof}
According to the direct sum decomposition \eqref{eq:dec-gz},
we can write, uniquely, $A(z)=h(z)+B(z)$,
where $h(z)\in\mf h$ and $B(z)\in\mf h^\perp$.
Note that $f+zE$ commutes with $f_1+zE$, since $[f_2,f_1]=[f_2,E]=0$ (cf. Definition \ref{20200518:def3}).
Hence $\mf h=\ker\ad (f_1+zE)$ and $\mf h^\perp=\im\ad(f_1+zE)$ are both $\ad(f+zE)$-invariant.
By definition, $f_1+zE$ is the semisimple part of $f+zE$,
and obviously $\ad (f_1+zE)$ is invertible when restricted to $\mf h^\perp$.
Hence $\ad(f+zE)|_{\mf h^\perp}:\mf h^\perp\to\mf h^\perp$ is invertible as well.
Therefore, there exists a unique $U(z)\in\mf h^\perp$ such that  $[f+zE,U(z)]=B(z)$.
The claim follows.
\end{proof}

If $E\in\mf g_k$, we extend the $\frac12\mb Z$-grading \eqref{eq:dec}
to $\mf g((z^{-1}))$ by letting $z$ have degree $-k-1$,
so that $f+zE$ and $f_1+zE$ are homogeneous of degree $-1$.
Then
\begin{equation}\label{decz}
\mf g((z^{-1}))=\widehat\bigoplus_{i\in\frac12\mb Z}\mf g((z^{-1}))_i\,,
\end{equation}
where $\mf g((z^{-1}))_i\subset\mf g((z^{-1}))$
is the space of homogeneous elements of degree $i$,
and the direct sum is completed by allowing infinite series in positive degrees,
cf. \cite[Lem.4.4(a)]{DSKV13}.
Note that, since $f_1+zE$ is homogeneous, we have the corresponding decompositions
of $\mf h$ and $\mf h^\perp$:
\begin{equation}\label{dech}
\mf h=\widehat\bigoplus_{i\in\frac12\mb Z}\mf h_i
\,\,\text{ and }\,\,
\mf h^\perp=\widehat\bigoplus_{i\in\frac12\mb Z}\mf h^\perp_i
\,.
\end{equation}

Recall that $\mc V(\mf p)$ is a commutative associative algebra
with derivation $\partial$.
Consider the Lie algebra 
$$
\widetilde{\mf g}=\mb F\partial\ltimes\big(\mc V(\mf p)\otimes\mf g((z^{-1}))\big)\,,
$$
where $\partial$ acts on the first factor.
For every $U(z)\in \mc V(\mf p)\otimes\mf g((z^{-1}))_{>0}\subset\widetilde{\mf g}$,
we have a well-defined automorphism of the Lie algebra $\widetilde{\mf g}$
given by $e^{\ad U(z)}$, cf. \cite[Lem.4.4(b)]{DSKV13}.
We extend the bilinear form $(\cdot\,|\,\cdot)$ on $\mf g((z^{-1}))$
to a map
\begin{equation}\label{20200625:eq1}
(\cdot\,|\,\cdot)
\,:\,\,
\big(\mc V(\mf p)\otimes\mf g((z^{-1}))\big)
\times
\big(\mc V(\mf p)\otimes\mf g((z^{-1}))\big)
\to
\mc V(\mf p)((z^{-1}))
\,,
\end{equation}
given by 
$(g\otimes a(z)|h\otimes b(z))
=
gh(a(z)|b(z))$.

Using the bilinear form $(\cdot\,|\,\cdot)$ on $\mf g$
we get the isomorphism $\mf p^*\simeq\mf m^\perp$.
Let $\{q_i\}_{i\in P}$ be a basis of $\mf p$, and let
$\{q^i\}_{i\in P}$ be the dual (with respect to $(\cdot\,|\,\cdot)$) basis of $\mf m^{\perp}$,
namely, such that $(q^j| q_i)=\delta_{ij}$.
We denote
\begin{equation}\label{eq:q}
q=\sum_{i\in P}q_i\otimes q^i
\,\in\mc V(\mf p)\otimes\mf m^{\perp}\,.
\end{equation}
The next result is a generalization of \cite[Prop.4.5]{DSKV13} to the case of an integrable triple associated to the nilpotent element
$f$.
\begin{proposition}\label{int_hier2_ds}
Let $f$ be a non-zero nilpotent element of the Lie algebra $\mf g$ which
admits an integrable triple $(f_1,f_2,E)$.
Then there exist unique formal Laurent series $U(z)\in\mc V(\mf p)\otimes\mf{h}^\perp_{>0}$
and $h(z)\in\mc V(\mf p)\otimes\mf{h}_{>-1}$ such that
\begin{equation}\label{L0_ds}
e^{\ad U(z)}(\partial+1\otimes(f+zE)+q)=\partial+1\otimes (f+zE)+h(z)\,.
\end{equation}
Moreover, an automorphism $e^{\ad U(z)}$, with $U(z)\in\mc V(\mf p)\otimes\mf{g}((z^{-1}))_{>0}$,
solving \eqref{L0_ds} for some $h(z)\in\mc V(\mf p)\otimes\mf{h}_{>-1}$,
is defined uniquely up to multiplication on the left by automorphisms of the form $e^{\ad S(z)}$,
$S(z)\in\mc V(\mf p)\otimes\mf h_{>0}$.
\end{proposition}
\begin{proof}
Let us write $U(z)=\sum_{i\geq\frac12}U_i(z)$,
where $U_i(z)\in\mc V(\mf p)\otimes\mf h^{\perp}_i,\,i\geq\frac12$,
and $h(z)=\sum_{i\geq-\frac12}h_i(z)$,
where $h_i(z)\in\mc V(\mf p)\otimes\mf h_i,\,i\geq-\frac12$,
and $\mf h_i$, $\mf h_i^\perp$ are defined by \eqref{dech}.
We determine $U_{i+1}(z)\in\mc V(\mf p)\otimes\mf h^{\perp}_{i+1}$
and $h_i(z)\in\mc V(\mf p)\otimes\mf h_i$, inductively on $i\geq-\frac12$,
by equating the homogeneous components of degree $i$ 
in each sides of equation \eqref{L0_ds}.
This amounts to solving an equation 
in $h_i(z)$ and $U_{i+1}(z)$ of the form
\begin{equation}\label{20200619:eq1}
h_i(z)+[1\otimes(f+zE),U_{i+1}(z)]=A(z)\,,
\end{equation}
where $A(z)\in\mc V(\mf p)\otimes\mf g((z^{-1}))_i$ is a (differential polynomial) 
expression involving all the elements $U_{j+1}(z)$ and $h_j(z)$ for $j<i$ (see the proof of \cite[Prop.4.5(a)]{DSKV13} for further details).
By Lemma \ref{lem:main} there exists a unique solution  $h_i(z)\in\mc V(\mf p)\otimes\mf h_i$ 
and $U_{i+1}(z)\in\mc V(\mf p)\otimes\mf h^\perp_{i+1}$, thus proving the first part of the proposition.

Next, let
$\widetilde U(z)\in\mc V(\mf p)\otimes \mf g((z^{-1}))_{>0},\,
\widetilde h(z)\in\mc V(\mf p)\otimes\mf h_{>-1}$ 
be some other solution of \eqref{L0_ds}:
$e^{\ad \widetilde{U}(z)}(\partial+1\otimes(f+zE)+q)=\partial+1\otimes(f+zE)+\widetilde{h}(z)$.
By the Baker-Campbell-Hausdorff formula,
there exists $R(z)=\sum_{i>0}^\infty R_i(z)\in\mc V(\mf p)\otimes \mf g((z^{-1}))_{>0}$ 
such that
$e^{\ad\widetilde U(z)}e^{-\ad U(z)}=e^{\ad R(z)}$.
According to the direct sum decomposition \eqref{eq:dec-gz},
we can write, uniquely, $R(z)=S(z)+T(z)$,
where $S(z)\in\mc V(\mf p)\otimes\mf h_{>0}$ and $T(z)\in\mc V(\mf p)\otimes\mf h^\perp_{>0}$.
To conclude the proof of the proposition we need to show that $T(z)=0$.
By construction, we have
\begin{equation}\label{eq:ciaociao}
\partial+1\otimes(f+zE)+\widetilde{h}(z)=e^{\ad R(z)}(\partial+1\otimes(f+zE)+h(z))\,.
\end{equation}
Comparing the terms of degree $-\frac12$ in both sides of the above equation,
we get
$$
[1\otimes(f+zE),T_{\frac12}(z)]
=
h_{-\frac12}(z)-\widetilde{h}_{-\frac12}(z)-[1\otimes(f+zE),S_{\frac12}(z)]\,.
$$
Since $\mf h$  and $\mf h^\perp$ are both $\ad(f+zE)$-invariant, we have that
$[1\otimes(f+zE),T_{\frac12}(z)]\in(\mc V(\mf p)\otimes\mf h_{-\frac12})\cap(\mc V(\mf p)\otimes\mf h_{-\frac12}^\perp)=0$.
Hence $T_{\frac12}(z)=0$, since $\ad(f+zE)|_{\mf h^\perp}:\mf h^\perp\to\mf h^\perp$ is invertible.
Let us assume, by induction, that $T_j(z)=0$ for all $j<i$.
Comparing the terms of degree $i-1$ in both sides of equation \eqref{eq:ciaociao},
we easily get that, using the fact that $\mf h$ is a Lie subalgebra, 
$[1\otimes(f+zE),T_{i}(z)]\in(\mc V(\mf p)\otimes\mf h^\perp_{i-1})\cap(\mc V(\mf p)\otimes\mf h_{i-1})=0$,
namely $T_{i}(z)=0$, as desired.
\end{proof}
The main result of this section is the following theorem,
which allows us to construct an integrable hierarchy of bi-Hamiltonian equations for classical affine $W$-algebras.
\begin{theorem}\label{final}
  Let $\mf g$ be a reductive Lie algebra with a non-degenerate symmetric invariant bilinear form $(\cdot\,|\,\cdot)$ extended to
  $\mf g((z^{-1}))$ by bilinearity. Let $f$
be a non-zero nilpotent element of $\mf g$, and
let $(f_1,f_2,E)$ be an integrable triple for $f$ 
and consider the decomposition \eqref{eq:dec-gz} of $\mf g((z^{-1}))$.
Let $U(z)\in\mc V(\mf p)\otimes\mf g((z^{-1}))_{>0}$ and $h(z)\in\mc V(\mf p)\otimes\mf h_{>-1}$
be a solution of equation \eqref{L0_ds}.
Let $\mf c(\mf h)$ be the center of $\mf h$, and let
\begin{equation}\label{20200625:eq2}
\mf a=\big\{a\in\mf c(\mf h)\,\big|\, [a,f_2]=0\big\}\subset\mf g((z^{-1}))
\,.
\end{equation}
For $a\in\mf a$, let
\begin{equation}\label{20200619:eq2}
\tint g_a(z)=\tint(1\otimes a|h(z))=\sum_{n\in\mb Z_{\geq 0}}\tint g_{a,n}z^{N-n}\,,
\end{equation}
where $N$ is the largest power of $z$ appearing in $\tint(1\otimes a|h(z))$
with non-zero coefficient,
and $(\cdot\,|\,\cdot)$ is defined in \eqref{20200625:eq1}.
Let $\mc W=\mc W(\mf g,f,E)\subset\mc V(\mf p)$ be the classical affine $W$-algebra with its
compatible PVA structures 
$\{\cdot\,_\lambda\,\cdot\}_{0}$ and $\{\cdot\,_\lambda\,\cdot\}_{\infty}$
introduced in Definition \ref{def:clw}.
Then 
$$
\mc A=\Span\{\tint g_{a,n}\mid a\in\mf a, n\in\mb Z_{\geq 0}\}
$$
is an infinite-dimensional abelian subalgebra of
$\mc W/\partial\mc W
\subset\mc V(\mf p)/\partial\mc V(\mf p)$
(with respect to both $0$ and $\infty$-Lie brackets),
defining a hierarchy of bi-Hamiltonian equations
\begin{equation}\label{eq:DShier}
\frac{dw}{dt_{a,n}}
=\{{g_{a,n}}_{\lambda}w\}_{0}\big|_{\lambda=0}
=\{{g_{a,n+1}}_{\lambda}w\}_{\infty}\big|_{\lambda=0}
\,,
\quad w\in\mc W\,,a\in\mf a\,,n\in\mb Z_{\geq 0}\,.
\end{equation}
Moreover, 
if $a\in\mf a$ is not central in $\mf g((z^{-1}))$,
then $\dim\Span\{\tint g_{a,n}\}_{n\in\mb Z_{\geq0}}=\infty$.
Consequently, the hierarchy \eqref{eq:DShier} is integrable.
\end{theorem}
\begin{proof}
First, we show that the results in \cite[Sec.4.5]{DSKV13} remain valid if $\mf h=\ker\ad(f_1+zE)$ and $a\in\mf a$.
Indeed, 
by the definition \eqref{20200625:eq2} of $\mf a$,
$[h(z),1\otimes a]=0$, $[f_1+zE,a]=0$, $[f_2,a]=0$,
so that $[\partial+1\otimes(f+zE)+h(z),1\otimes a]=0$.
Applying the automorphism $e^{-\ad U(z)}$
and using equation \eqref{L0_ds},
we thus get
$$
[\partial+1\otimes (f+zE)+q,e^{-\ad U(z)}(1\otimes a)]=0
\,,
$$
which is the analogue of \cite[Cor.4.6]{DSKV13}.
The proofs of \cite[Lem.4.7, Lem.4.8, Thm.4.9]{DSKV13} go through unchanged
in the present setting.
As a result,
we have that the local functionals $\tint g_{a,n}\in\mc V(\mf p)/\partial\mc V(\mf p)$, $n\in\mb Z_{\geq 0}$, 
defined by equation \eqref{20200619:eq2}, solve the Lenard-Magri recurrence relation \eqref{eq:LM}
and $\{\tint g_{a,0},\tint w\}_{\infty}=0$, for every $w\in\mc W$. Hence, by Lemma \ref{20200609:lem1},
all $\int g_{a,n}$ commute in the Lie algebra $\mc W/\partial\mc W$.

Next, by Proposition \ref{int_hier2_ds} and the proofs of \cite[Lem.4.10, Prop.4.11]{DSKV13}, it follows that
the Laurent series $\tint g_a(z)$ defined by \eqref{20200619:eq2}
is independent of the choice of 
$U(z)\in\mc V(\mf p)\otimes\mf g((z^{-1}))_{>0}$, $h(z)\in\mc V(\mf p)\otimes\mf h_{>-1}$,
solving equation \eqref{L0_ds}, and that its coefficients $\tint g_{a,n}$, $n\in\mb Z_{\geq 0}$, lie in $\mc W/\partial\mc W$.
Hence, by Lemma \ref{20200609:lem1}, we have that $\mc A$ is an abelian subalgebra, and, by the discussion in Section \ref{sec:3.1},
we get the hierarchy of bi-Hamiltonian equations \eqref{eq:DShier}.

Finally, we are left to show that if $a$ is not in the center of $\mf g((z^{-1}))$, then the local functionals
$\tint g_{a,n}$, $n\in\mb Z_{\geq 0}$, span an infinite-dimensional subspace. 
This follows verbatim from the results in \cite[Sec.4.7]{DSKV13} noticing that $\ad(f+zE)$ and $\ad a$ commute.
Since $f_1+zE\in\mf a$, we see that $\mc A$ is infinite-dimensional, as claimed.
\end{proof}
\begin{remark}\label{rem:victor}
Theorem \ref{final} holds in the more general setting of an arbitrary finite-dimensional Lie algebra $\mf g$
with a non-degenerate symmetric invariant bilinear form $(\cdot\,|\,\cdot)$,
an $\mf{sl}_2$-triple $\{e,h,f\}$ in $\mf g$,
and an integrable triple $(f_1,f_2,E)$ for $f$.
\end{remark}


\end{document}